\documentclass[journal]{IEEEtran}%
\usepackage{cite}
\usepackage{amsmath,amssymb,amsfonts}
\usepackage{algorithmic}
\usepackage{graphicx}
\usepackage{textcomp}
\usepackage{xcolor}
\usepackage{xfrac}
\usepackage{amsmath}
\usepackage{amsfonts}
\usepackage{amssymb}%
\providecommand{\U}[1]{\protect\rule{.1in}{.1in}}
\def\BibTeX{{\rm B\kern-.05em{\sc i\kern-.025em b}\kern-.08em
T\kern-.1667em\lower.7ex\hbox{E}\kern-.125emX}}

\usepackage{amsmath}
\usepackage{amsthm}
\usepackage{amssymb}

\usepackage{hyperref}
\usepackage{verbatim} 

\usepackage{graphicx}
\usepackage{caption}
\usepackage{subcaption}
\usepackage{ textcomp } 
\usepackage{color}
\usepackage{enumerate}
\usepackage{mathrsfs}
\usepackage{multicol}
\usepackage{amsopn}

\usepackage{multicol}

\usepackage{algorithm}
\usepackage{algorithmic} 

\usepackage{graphics}
\usepackage{tikz}
\usepackage{xfrac}
\usepackage{cite}
\usepackage{epstopdf}
\usepackage{pdfpages}
\usepackage{dsfont}
\usepackage{amsfonts}

\newcommand{\C}{\mathbb{C}}

\newcommand{\calN}{\mathcal{N}}

\newcommand{\calX}{\mathcal{X}}
\newcommand{\calY}{\mathcal{Y}}

\newcommand{\abs}[1]{\left\vert #1 \right\vert}
\newcommand{\norm}[1]{\Vert #1 \Vert}

\newcommand{\set}[1]{\left\lbrace #1\right\rbrace}
\newcommand{\sse}{\subseteq}
\newcommand{\sprod}[1]{\left\langle #1 \right\rangle}

\newcommand{\prb}[1]{\mathbb{P}\left( #1 \right)}

\usepackage{dsfont}

\newcommand{\geqsim}{\gtrsim}

\DeclareMathOperator{\supp}{supp}
\DeclareMathOperator{\id}{id}

\newcommand{\argmin}{\mathop{\mathrm{argmin}}}

\newtheorem{lem}{Lemma}
\newtheorem{prop}[lem]{Proposition}

\newtheorem{theo}[lem]{Theorem}

\newtheorem*{ass}{Attack model}

\newtheorem{rem}{Remark}

\newtheorem*{inftheo}{Main result (informal)}

\newenvironment{2coltable}{\begin{table*}[h]}{\end{table*}}
\newenvironment{figure2}{\begin{figure*}[h]}{\end{figure*}}
\newenvironment{algorithm2}{\begin{algorithm*}[ptb]}{\end{algorithm*}}

\theoremstyle{definition}

\numberwithin{lem}{section}

\usepackage{etoolbox}
\let\bbordermatrix\bordermatrix
\patchcmd{\bbordermatrix}{8.75}{4.75}{}{}
\patchcmd{\bbordermatrix}{\left(}{\left[}{}{}
\patchcmd{\bbordermatrix}{\right)}{\right]}{}{}

\definecolor{axel}{rgb}{0,.6,.8}
\definecolor{gerhard}{rgb}{1,.8,.1}

\begin{document}

\title{Perfectly Secure Key Agreement Over a Full Duplex Wireless Channel}


\author{
\IEEEauthorblockN{Gerhard Wunder}
\IEEEauthorblockA{\textit{Freie Universität Berlin} \\
Berlin, Germany \\
g.wunder@fu-berlin.de \\}
\thanks{Parts of these results were presented at the IEEE CommNet 2023.}
\and
\IEEEauthorblockN{Axel Flinth}
\IEEEauthorblockA{\textit{Umeå University} \\
Umeå, Sweden \\
axel.flinth@umu.se\\}
\and
\IEEEauthorblockN{Daniel Becker}
\IEEEauthorblockA{\textit{Freie Universität Berlin} \\
Berlin, Germany \\
daniel.becker@fu-berlin.de \\}
\and
\IEEEauthorblockN{Benedikt Groß}
\IEEEauthorblockA{\textit{Freie Universität Berlin} \\
Berlin, Germany \\
benedikt.gross@fu-berlin.de}
}

\maketitle

\begin{abstract}
Secret key generation (SKG) between authenticated devices is a pivotal task
for secure communications. Diffie-Hellman (DH) is de-facto standard but not
post-quantum secure. In this paper, we shall invent and analyze a new security
primitive that is specifically designed for WPAN. For WPAN, wireless
channel-based SKG has been proposed but was not widely deployed due to its
critical dependence on the channel's entropy which is uncontrollable. We
formulate a different approach: We still exploit channel properties but mainly
hinge on the reciprocity of the wireless channel and not on the channel's
entropy. The radio advantage comes from the use of full duplex communication.
We show that in this situation both legitimate parties can agree on a common
secret key even without ever probing the channel at all. At the core is a new
bisparse blind deconvolution scheme for which we prove correctness and
information-theoretic, i.e. perfect, security. We show that, ultimately, a
secret key can be extracted and give a lower bound for the number of secret
key bits which is then verified by experiments.
\end{abstract}


\section{Introduction}

Diffie-Hellman (DH) is one of the most used key exchange algorithms. However,
with the aspiring threat of quantum algorithms that solve the underlying
mathematical problems in polynomial time, the need for post-quantum secure key
exchange protocols is evident \cite{NIST2022_PQC}. In this paper, we shall
invent and analyze a new security primitive, termed Full Duplex - Bisparse
Blind Deconvolution, in short FD-BBD, that is specifically designed for
wireless personal area networks (WPANs). One potential application scenario in
WPANs is the so-called \emph{close talker scenario}\ \cite{Pierson19,Khan21}
where two legitimate devices (Alice and Bob) come ad hoc together and want to
securely communicate in the presence of a potential eavesdropper (Eve). They
do not know each other but authenticate themselves through their proximity
\cite{Khan21}. For secure communication, a session key is then generated
through FD-BBD. Notably, the near-field channel between Alice and Bob is
reciprocal and even highly frequency-selective \cite{Pierson19_2}. In
far-field, the channels from Alice or Bob to Eve are quite similar when the
legitimate devices are sufficiently close (within centimeters). Both is proved
to give Alice and Bob the critical radio advantage over Eve in FD-BBD.

FD-BBD generally falls in the category of \emph{Physical Layer Security (PLS)}
\cite{Mukherjee16}. While there is no clear definition, PLS genuinely
incorporate the principles of electromagnetic wave propagation into account,
i.e. superposition, channel reciprocity, fading. Principles that surprisingly
uphold not only in far-field but also in the near-field by virtue of Maxwells
equations \cite{Tang2021}. The standard PLS approach is to measure the
reciprocal and exclusive channel $h_{AB}$ between Alice and Bob, each using
probing signals, from which a secret key is then extracted. However, since the
available entropy in $h_{AB}$ is unpredictable and often far too little, this
scheme has not been widely deployed \cite{Zhang16}. In fact, to the authors'
knowledge it has not been implemented in any proprietary system or standard.

\begin{figure2}
\centering\includegraphics[width=1\linewidth]{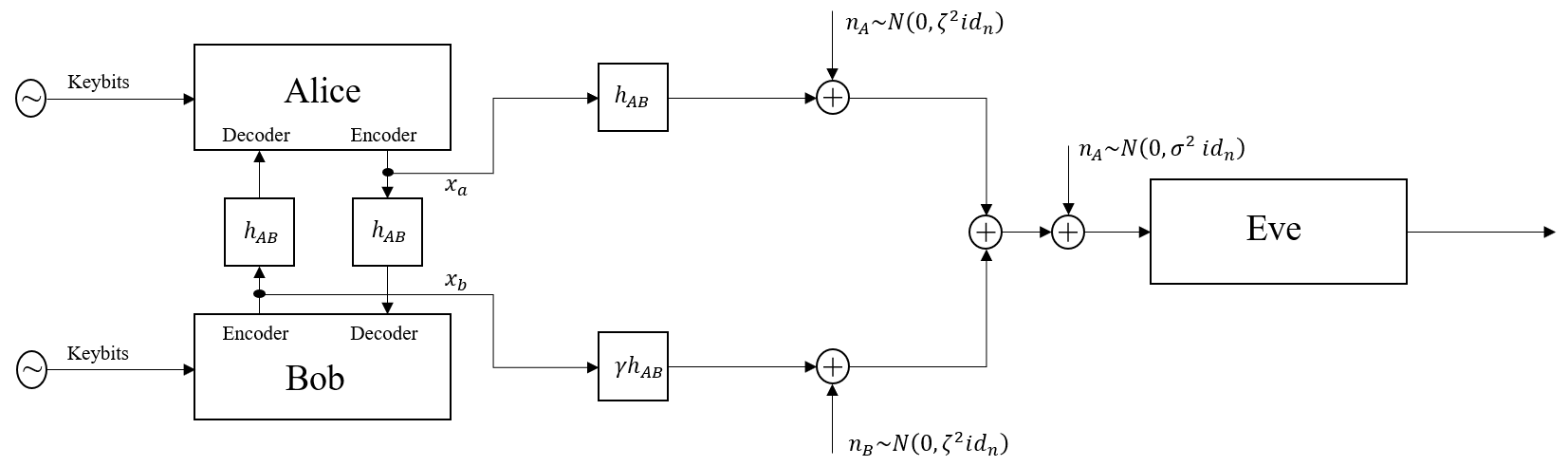}\caption{\emph{The close
talker scenario}: The individual signals $x_{A},x_{B}$ go through the reciprocal channel
$h_{AB}$ for Alice and Bob, while outside some 'trust' region, Eve can only
see the superposition of both through $h_{AE},h_{BE}$ due to full duplex
transmission. Here, $h_{AE}$ and $h_{BE}$ are scaled and noisy versions of each other.} 

\label{closetalker}
\end{figure2}

In sharp contrast to this traditional approach, the way how FD-BBD generates
the secret key is entirely different: The channel $h_{AB}$ is neither the
primary source for its generation, nor is the channel ever explicitly measured
through channel probing at all. Instead, Alice and Bob generate local signals
$\beta_{B},\beta_{A}$ which are on/off signals that carry the information bits
in their random supports $\sigma_{A},\sigma_{B}$. The signals are then mapped
onto a (compressive) linear code and sent through the channel. Again, no
probing signals are needed, just rough synchronization between Alice and Bob
on frame level is required. Then, to actually generate the secret key, we
employ a new \emph{bisparse blind deconvolution} algorithm. The signals do not 
not individually recovered: Instead, the secret key is obtained from the product of
the discrete Fourier transforms (DFTs) of suitable upsampled version of
$h_{AB},\beta_{B},\beta_{A}$. The upsampling is a by-product of the specific
hierarchical way our bisparse blind deconvolution algorithm works. -- Eve, on
the other hand, can only see the superposition of the signals due to the full
duplex mode, so it can only recover the superposed signals. We will spend a
great deal to prove that such superposition does not reveal anything about the
secret key above, which is essentially to show that any possible support
configuration of $\beta_{B}+\beta_{A}$ leads to a different secret (injective
mapping). Ultimately, we prove correctness and information-theoretic, i.e.
perfect, security of the FD-BBD scheme and give a lower bound on the number of
secret key bits. Eventually, to avoid confusion, we would like to emphasize that:

\begin{enumerate}
\item the eavesdropper knows the used linear coding matrix, i.e. it is
\emph{not} kept secret,

\item advanced coding is required since Alice and Bob communicate \emph{ad
hoc} without channel estimation, and, 

\item for the way we create the secret, full duplex alone would be not
sufficient for perfect secrecy.
\end{enumerate}

The main result is now summarized as follows:

\begin{inftheo}
With bisparse blind deconvolution, Alice and Bob can both recover a common
secret of the form
\[
\mathfrak{c}=\widehat{h_{AB}^{\Uparrow}}\cdot\widehat{\beta_{B}^{\uparrow}%
}\cdot\widehat{\beta_{A}^{\uparrow}}\in\mathbb{C}^{\mu n}%
\]
where $\widehat{\cdot}$ denotes the DFT and multiplication is element-wise
from which the secret key is obtained (see below for the definitions of
$\mu,n\in\mathbb{N}$). Herein $h_{AB}^{\Uparrow},\beta_{B}^{\uparrow}%
,\beta_{A}^{\uparrow}$ are suitably upsampled versions of reciprocal channel
and local signals of Alice and Bob, respectively, the latter each with random
supports $\sigma_{A},\sigma_{B}$ of cardinality $k$. The number of information
bits is hence $\log_{2}\binom{2k}{k}$ bits. Moreover, if Eve has access to
$\Sigma_{AB}=\sigma_{A}\cup\sigma_{B}$ the mapping
\[
\sigma_{A}=\supp\beta_{A}\rightarrow\mathfrak{c}\text{\quad or\quad}\sigma
_{B}=\supp\beta_{B}\rightarrow\mathfrak{c}%
\]
is injective.
Consequently, the resulting uncertainty about $\sigma_{A}$ (or $\sigma_{B}$)
is preserved even if the sum set $\Sigma_{AB}$ is known to Eve.
\end{inftheo}

We note that there is some (remote) similarity of the FD-BBD protocol to DH
key exchange by identifying $\widehat{\beta_{B}^{\uparrow}}\triangleq
g^{x},\widehat{\beta_{A}^{\uparrow}}\triangleq g^{y},\widehat{h_{AB}%
^{\Uparrow}}\triangleq1$, where the reverse operation (discrete logarithm
$g^{(.)}$ in a strong cryptographic field using generator $g$) is one-way (but
not perfectly secure) and where exponentiation (like the element-wise
multiplication in DFT domain) is commutative. The main purpose of this paper
is to show that the above mapping is correct and securely encapsulates the key
within the close talker scenario.

\section{Signal model}

\subsection{Signal design and attacker model}

\label{sec:wdh_scheme}

Alice and Bob first generate random signals $\beta_{A}\in\C^{n}$ and
$\beta_{B}\in\C^{n}$ as follows: The $\beta$'s are supported on uniformly
random sets $\sigma\sse [n]:=\{0,...,n-1\}$ of cardinality $k$, with non-zero
values $\alpha_{k}$ drawn according to some fixed distribution (to be defined
later). That is:%
\[
\beta_{A/B}=\sum_{j\in\lbrack k]}\alpha_{j}e_{\sigma_{A/B}\left(  j\right)
},
\]
where for $j\in\lbrack n]$, $e_{j}\in\C^{n}$ denotes the $j$-th unit
vector\footnote{All numbering of set elements starts from 0 to set size minus
one.}. After generating their signals, Alice and Bob both apply a common
public codebook $Q\in\C^{\mu\times n}$ to their respective signals, and
transmit
\[
x_{A}=Q\beta_{A}\in\C^{\mu},\hspace{1cm}x_{B}=Q\beta_{B}\in\C^{\mu}%
\]
simultaneously \emph{in full duplex mode} over a common wireless channel.
Here, $\mu\gg1$ is the signal space dimension. We require $\mu\leq n$, with
$\mu$ typically much smaller than $n$ (compression). The resulting signal is
then%
\begin{align*}
y_{A}  &  =h_{A}\ast Q\beta_{A}+h_{B\rightarrow A}\ast Q\beta_{B}+n_{B}%
\in\C^{\mu}\\
y_{B}  &  =h_{B}\ast Q\beta_{B}+h_{A\rightarrow B}\ast Q\beta_{A}+n_{A}%
\in\C^{\mu}%
\end{align*}
where $h_{A},h_{B}$ denote the self-interferences, $h_{A\rightarrow
B},h_{B\rightarrow A}$ are the wireless channels, and $n_{A},n_{B}$ is the
noise, which we assume to be $\calN(0,\sigma^{2}\id_{\mu})$-distributed. As
discussed, the channel is reciprocal (even in near-field), which means that
$h_{A\rightarrow B}=h_{B\rightarrow A}=h_{AB}$. Moreover, the
self-interference terms are considered small and negligible
\cite{Fritschek2017_CNS} which is a common assumption. With these assumptions
we get
\begin{align*}
y_{A}  &  =h_{AB}\ast Q\beta_{B}+n_{B}\in\C^{\mu}\\
y_{B}  &  =h_{AB}\ast Q\beta_{A}+n_{A}\in\C^{\mu}.
\end{align*}
For the channel $h_{AB}$, we will use a standard random model for the recovery
analysis: It consists of $s$ distinct paths, each with a normally distributed
complex gain and uniformly random position. That is,
\[
h_{AB}=\sum_{k\in\lbrack s]}\gamma_{k}e_{i_{k}}%
\]
for $\{\gamma_{k}\}_{k\in\lbrack s]}$ independent Gaussians and $\{i_{k}%
\}_{k\in\lbrack s]}$ independently distributed on $[\mu]$. Finally, $h_{AB}$
is independent of $\beta_{A}$ and $\beta_{B}$. However, we emphasize that it
is not essential for the security analysis.

Finally, let us include the passive eavesdropper Eve carrying out a \emph{key
recovery attack}. Eve observes
\[
y_{E}=h_{A\rightarrow E}\ast Q\beta_{A}+h_{B\rightarrow E}\ast Q\beta
_{B}+n_{E}\in\C^{\mu}%
\]
where channels $h_{A\rightarrow E},h_{B\rightarrow E}$ are in general
different as well as (typically) independent of $h_{AB}$ and
frequency-selective. In this paper, we will consider attacker model model,
motivated by the described close talker scenario.

\begin{ass}
\label{ass:filter} The eavesdropper's channels $h_{A\rightarrow E}%
,h_{B\rightarrow E}$ are almost equal to the reciprocal channel $h_{AB}$,
subject to a scaling factor. More concretely, there exists a $\gamma>0$ and
vectors $n_{A}$, $n_{B}$ such that (after complex rotation)
\[
h_{A\rightarrow E}=h_{AB}+n_{A},h_{B\rightarrow E}=\gamma h_{AB}+n_{B}.
\]
$n_{A}$ and $n_{B}$ are assumed to have the same support as $h_{AB}$, with
normal Gaussian distributed values. That is, $n_{A},n_{B}\sim\calN(0,\varsigma
^{2}\id_{\supp h_{AB}})$ where $\id_{n}$ is the identity operator in $n$ dimensions.

Under these assumptions, Eve receives%
\begin{equation}
y_{E}=h_{AB}\ast Q(\beta_{A}+\gamma\beta_{B})+n_{A}\ast Q\beta_{A}+n_{B}\ast
Q\beta_{B}+n_{E}. \label{eq:erweitert}%
\end{equation}
Hence, at least when $\varsigma=\sigma=0$, Eve can recover each component of
$\beta_{A}+\gamma\beta_{B}$. Specifically, she can recover all information in the case $\gamma=0$. Thus, the distinguishing factor $\gamma>0$ quantifies the extent of
information that Eve can gather after reception. Obviously, Eve must still
somehow infer the (possibly static) channel from Alice and Bob (over time)
though, which gives additional security.
\end{ass}

Figure~\ref{closetalker} depicts the signal and attack model.



\begin{algorithm2}
\caption{HiHTP \label{alg:hihtp}}
\begin{algorithmic}
\REQUIRE Linear operator $C$, measurement $y$, and sparsity parameters $(s,k)$
\STATE Initialize $W$ \REPEAT
\STATE $W_{i+\sfrac{1}{2}} \leftarrow W_{i} + C^{*}(y-CW_{i})$ \STATE  $\Omega
_{i}$ $=$ $\supp$(best $(s,k)$-sparse approximation of $W_{i+\sfrac{1}{2}}$)
\STATE $W_{i} = \argmin_{\supp W \sse \Omega_{i}} \norm{y - C(W)}^{2}$
\UNTIL{stopping criterion is met} \ENSURE Approximate $(s,k)$-sparse solution
of $y=C(W)$
\end{algorithmic}
\end{algorithm2}

\subsection{Hierarchically sparse blind deconvolution}

Using the well-known lifting trick \cite{candes:rip2008,candes2015phase}, the
bi-linear equation can be transformed into a linear one as (missing subscripts
$A,B$)%
\[
y=B\mathrm{vec}\left(  hx^{T}\right)  +n
\]
Here, $B\in\left(  0,1\right)  ^{n\times n^{2}}$ is a suitable matrix with
$(B)_{i,jn+k}=\delta_{i,j+k\mod n}$, which is composed as%
\[
B=\left(
\begin{array}
[c]{ccccc}%
10...0 & 0...01 & 0...10 & ... & 01...0\\
01...0 & 10...0 & 0...01 & ... & ...\\
... & ... & ... & ... & ...\\
... & ... & ... & ... & ...
\end{array}
\right)
\]
and $n$ is the Gaussian noise vector. The sparse signal model $x=Q\beta$ with
the random coding matrix $Q\in\mathbb{C}^{\mu\times n}$ and $k$-sparse binary
vector $\beta\in\{0,1\}^{n}$ of length $n$ can be incorporated in the
formulation to yield
\[
y=B(\id_{n}\otimes Q)\mathrm{vec}\left(  h\beta^{T}\right)
\label{eq:A-operator}%
\]
By this procedure, the blind deconvolution problem of recovering $h$ and
$\beta$ from the measurement $y$ is turned into a matrix recovery problem in
$\beta h^{T}$, or equivalently of the tensor $h\otimes\beta\in\C^{\mu}%
\otimes\C^{n}$, from linear measurements $C:\C^{\mu}\otimes\C^{n}%
\rightarrow\C^{\mu}$, where $C$ is defined through
\begin{equation}
C(h\otimes\beta)=B(\id_{n}\otimes Q)\mathrm{vec}\left(  h\otimes\beta\right)
.\label{eq:C}%
\end{equation}
The factors $h\in\C^{\mu}$ and $\beta\in\C^{n}$ can be obtained from $X$ as
the first left and right singular vectors of the SVD of $X$. In this notation,
the observations of Alice, Bob and Eve, respectively, are given by
\begin{align*}
y_{A} &  =C(h_{AB}\otimes\beta_{B})+n_{B},\\
y_{B} &  =C(h_{AB}\otimes\beta_{A})+n_{A},\\
y_{E} &  =C(h_{A\rightarrow E}\otimes\beta_{A}+h_{B\rightarrow E}\otimes
\beta_{B}))+n_{E}.
\end{align*}
As outlined in the previous section, both the legitimate users and the
adversary want to deconvolve $h\ast Q\beta$ to recover $h$ and $\beta$. This
amounts to recovering $h\otimes\beta$ from the linear measurements
$C(h\otimes\beta)$. To do this, we can utilize the structure of the matrix
$X=h\otimes\beta$: It contain non-zero entries only in $s$ columns (since $h$
does), and that each of these columns in themselves are $k$-sparse (since $b$
is). This type of signals are known as $(s,k)$-sparse signals(or more
informally hisparse signals). They can be recovered from the measurements $y$
using the so-called Hierarchical Thresholding Pursuit (HiHTP) (Algorithm
\ref{alg:hihtp}). (We refer to \cite{eisert2021hierarchical} for a more
thorough introduction of hierarchical sparsity).

Each iteration of the HiHTP algorithm consists of three steps: $(i)$ Take a
gradient descent step from the current iterate $W_{i}$ with respect to the
quadratic loss $\ell(W)=\norm{y-CW}^{2}$. $(ii)$ Project (\emph{hierarchically
threshold}) the resulting matrix $W_{i+\sfrac{1}{2}}$ is onto the set of
$(s,k)$-sparse signals, and record its support $\Omega_{i}$. $(iii)$ Minimize
the loss $\ell$ among all matrices supported on $\Omega_{i}$ to determine
$W_{i+1}$. This technique of subsequent gradient steps, projections and local
least-squares problems is more generally applicable in the framework of
\emph{model-based compressed sensing} \cite{baraniuk2010model}. What makes the
HiHTP algorithm special is that the projection step
is computationally cheap: To project a matrix $W$ onto the set of
$(s,k)$-sparse signals, one first determines the best $k$-sparse approximation
$\hat{w}_{\ell}$ of each column $w_{\ell}$ (this is simple: just set all
entries but the $k$ largest ones to zero). Next, the norms of the approximated
columns $\hat{w}_{\ell}$ are calculated. The best approximation of $W$ is then
obtained by setting all columns $\hat{w}_{\ell}$ but the ones with the $s$
largest norms to zero (see also Figure \ref{fig:proj}). Thus, the projection
can be calculated in $O(\dim(W))$ time, which is the same complexity as
calculating the best $sk$-sparse approximation of $W$ (which would be needed
for the simpler HTP algorithm, which is oblivious to the hierarchical nature
of the sparsity of the signal to be recovered.). \begin{figure}[ptb]
\centering
\includegraphics[width=.8\linewidth]{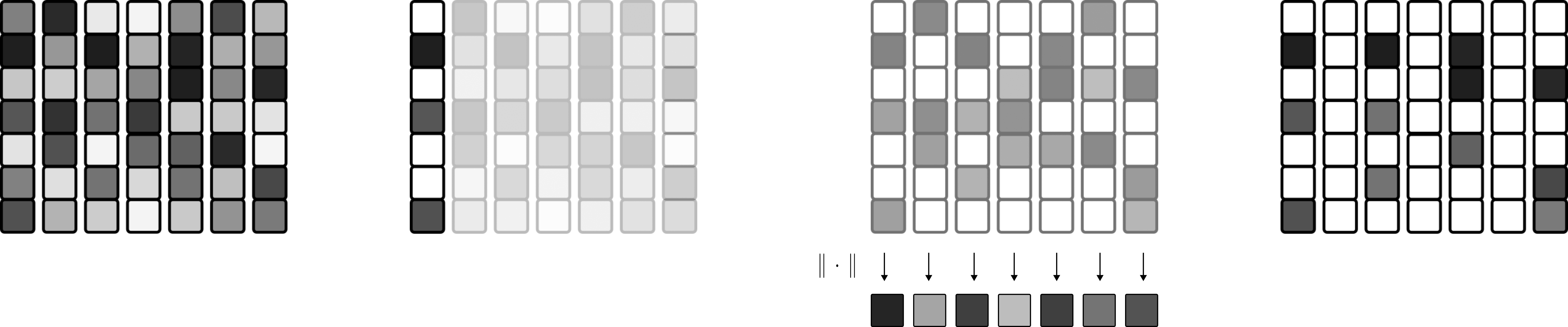} \caption{The two
phases of projection onto the set $(s,k)$-sparse signals. First, each
individual column is projected to its best $k$-sparse approximation. Then, the
$s$ largest approximations in $\ell_{2}$-norm are kept. }%
\label{fig:proj}%
\end{figure}

The algorithm provably converges towards any $(s,k)$-sparse ground truth, in a
robust fashion, as soon as the linear operator $C$
has the
\emph{HiRIP} property \cite{HiHTP}. The latter means that the operator acts
almost isometrically on the set of hierarchically sparse signals, i.e.
$\norm{C(W)}\approx\norm{W}$ for all
hisparse $W$. In Section~\ref{sec:correctness}, we will discuss a viable
randomized choice of $Q$ which implies that $C$ has the HiRIP property with
high probability -- which in turn implies that our algorithm succeeds with
high probability.

\subsection{Key generation}

For the key generation, both Alice and Bob use the blind deconvolution solver
to recover $\mathrm{vec}(h\otimes\beta_{B})$ (Alice) and $\mathrm{vec}%
(h\otimes\beta_{A})$ (Bob). Once Alice has determined the value of
$h\otimes\beta_{B}$, she can (since she knows her signal $\beta_{A}$)
calculate the value of
\begin{equation}
\mathfrak{c}\left(  \beta_{A},\beta_{B},h\right)  =\widehat{\mathrm{vec}%
(h\otimes\beta_{B})}\cdot\widehat{\mathrm{vec}(e_{0}^{\mu}\otimes\beta_{A}%
)},\label{eq:conv1}%
\end{equation}
where $\widehat{\cdot}$ again denotes the DFT. Similarly, Bob can calculate
\begin{equation}
\mathfrak{c}^{\prime}\left(  \beta_{B},\beta_{A},h\right)  =\widehat
{\mathrm{vec}(h\otimes\beta_{A})}\cdot\widehat{\mathrm{vec}(e_{0}^{\mu}%
\otimes\beta_{B})}\label{eq:conv2}%
\end{equation}
where we write here $e_{i}^{\mu}$ to mark that the $i$-th unit vector is in
$\mu$ dimensions. We claim that both terms \eqref{eq:conv1}, \eqref{eq:conv2}
represent a common secret for Alice and Bob.

To prove this, let us introduce two 'lifting' operations on vectors
$h\in\C^{\mu}$ and $\beta\in\C^{n}$:%
\[
h^{\Uparrow}=\sum_{k\in\lbrack\mu]}h_{k}e_{k}^{n\mu},\quad\beta^{\uparrow
}=\sum_{k\in\lbrack n]}\beta_{k}e_{k\mu}^{n\mu}.
\]
First, we show that tensor $h\otimes\beta$ can be reinterpreted as the
convolution of $h^{\Uparrow}$ with $\beta^{\uparrow}$. For this, it is not
hard to prove that in any dimension,
\[
e_{k}\ast e_{j}=e_{k+j}%
\]
by a direct calculation (or alternatively by arguing that $\widehat{(e_{k}\ast
e_{j})}=\widehat{e_{k}}\cdot\widehat{e_{j}}=\widehat{e_{k+j}}$). This implies
that
\begin{align*}
h^{\Uparrow}\ast\beta^{\uparrow} &  =\sum_{j\in\lbrack\mu]}\sum_{k\in\lbrack
n]}h_{j}e_{j}^{n\mu}\ast\beta_{k}e_{k\mu}^{n\mu}\\
&  =\sum_{j\in\lbrack\mu]}\sum_{k\in\lbrack n]}h_{j}\beta_{k}e_{j+k\mu}^{n\mu
}\\
&  =\mathrm{vec}(h\otimes\beta)
\end{align*}
as claimed. It is now simple to derive the equality of the keys. Taking the
Fourier-transformation yields
\[
\widehat{\mathrm{vec}(h\otimes\beta)}=\widehat{h^{\Uparrow}}\cdot
\widehat{\beta^{\uparrow}}.
\]
By in particular setting $h=e_{0}^{\mu}$, for which $h^{\Uparrow}=e_{0}^{n\mu
}$, we get
\[
\widehat{\mathrm{vec}(e_{0}^{\mu}\otimes\beta)}=\widehat{e_{0}^{n\mu}}%
\cdot\widehat{\beta^{\uparrow}}=\widehat{\beta^{\uparrow}}.
\]
By setting $h=h_{AB}$, and $\beta$ equal to $\beta_{A}$, or $\beta_{B}$,
respectively, in \eqref{eq:conv1}, \eqref{eq:conv2}, we now have indeed proved:

\begin{prop}
Alice's and Bob's secrets coincide:
\[
\mathfrak{c}=\mathfrak{c}^{\prime}=\widehat{h_{AB}^{\Uparrow}}\cdot
\widehat{\beta_{A}^{\uparrow}}\cdot\widehat{\beta_{B}^{\uparrow}}.
\]

\end{prop}

Now then the actual key is obtained from quantizing $\mathfrak{c}%
,\mathfrak{c}^{\prime}$ and hashing the outcomes. Moreover, the procedure can
be repeated \ $n$ rounds. We have now collected all the necessary ingredients
for the key generation. A summary of our protocol can be found in Table
\ref{tab:WDH_table}. \begin{2coltable}
\centering\resizebox{.7\linewidth}{!}
{
\begin{tabular}{lcl}
\hline
{\bf Alice}: Input $n,Q$ & & {\bf Bob}: Input $n,Q$ \\
\hline
{\bf For round i to m} & & {\bf For round i to m}\\
\ Choose random $\beta_A$ & & \ Choose random $\beta_B$\\
\ Send $x_A=Q\beta_A$ & $\longleftrightarrow$ & \ Send $x_B=Q\beta_B$\\
\ Measure $h* Q\beta_B$ & $\longleftrightarrow$ & \ Measure $h*Q\beta_A$ \\
\ Compute $\mathfrak{a_1}=h \otimes \beta_B$ && \ Compute  $\mathfrak{b_1}=h \otimes \beta_A$ \\
\ Compute $\mathfrak{a_2}=e_{0}^{\mu}\otimes\beta_{A}$ && \ Compute  $\mathfrak{b_2}=e_{0}^{\mu}\otimes\beta_{B}$ \\
\ Compute   $\mathfrak{c} = \widehat{\mathrm{vec}(\mathfrak{a}_1)}\cdot\widehat{\mathrm{vec}(\mathfrak{a}_2)}$
&& \ Compute $\mathfrak{c}^{\prime}= \widehat{\mathrm{vec}(\mathfrak{b}_1)}$
$\cdot\widehat{\mathrm{vec}(\mathfrak{b}_2)}$ \\
\ $\mathfrak{c_i^m}$ = $\mathfrak{c}$ & & \ $\mathfrak{c_i^m}^{\prime}$ = $\mathfrak{c}^{\prime}$\\
{\bf End for}& & {\bf End for}\\
Quantize and hash $\mathfrak{c^m}$ & & Quantize and hash $\mathfrak{c^m}^{\prime}$\\
\hline
\end{tabular}} \caption{WDH key agreement scheme ($\longleftrightarrow$ means full duplex transmission)}%
\label{tab:WDH_table}%
\end{2coltable}

\section{Security analysis}

\label{sec:correctness}

\subsection{Correctness}

First, let us use theory about bisparse signal retrieval to argue that $Q$ can
be chosen in such a way that Alice and Bob provably can solve their blind
deconvolution problems.

In the notation used here, we assume that the codebook $Q$ can be written as
$UA$, where $A\in\C^{m,n}$ is a matrix with a $k$-restricted isometry
property, and $U\in\C^{\mu,m}$ is a matrix with i.i.d. subgaussian entries. As
long as the dimension $\mu$ is sufficiently large, $Q$ will have the HiRIP,
which in turn shows that HiHTP succeds. Concretely, combining the results of
\cite{BCHTP} and \cite{eisert2021hierarchical}, we have the following.

\begin{theo}
Let $\varepsilon>0$. Suppose that the entries of the matrix $Q\in\C^{\mu,n}$
are i.i.d Gaussian with expected value $0$ and variance $\sfrac{1}{\mu}$.
Further assume that
\begin{equation}
\mu\geq K(s\log(s)^2\log(\mu)\log(\mu n)+sk\log(n))\cdot\max(1,\log(\varepsilon^{-1}))\,
\label{eq:mubound}%
\end{equation}
where $K$ is a constant independent of $s,\mu,k,n$ and $\varepsilon$. Then,
with a probability at least $1-\epsilon$, the following holds: For any
$s$-sparse $h$ and $k$-sparse $\beta$, the iterates $W_{k}$ of the
HiHTP-algorithm applied with $y=C(h\otimes\beta)+n$ and $C$ defined as in
\eqref{eq:C} fulfill
\[
\norm{W_j - h \otimes \beta}\leq\rho^{j}\norm{W_0-h\otimes \beta}+\theta
\norm{n},
\]
where $0\leq\rho<1$ and $\theta>0$ are constants only depending on $K$. In
other words, if the number of measurements fulfill \eqref{eq:mubound}, the
HiHTP will with high probability recover every $(s,k)$-sparse vector in a
robust fashion.
\end{theo}

The complexity bound \eqref{eq:mubound} is obviously not sample optimal --
counting linear degrees of freedom suggests that the problem should rather be
solvable using only $s+k$ measurements. This is however not as crucial as it
may seem: The WDH scheme does not rely on utilizing any specific algorithm for
the sparse blind deconvolution, and we may choose any of the many that are
available (see e.g.
\cite{li2019rapid,lee2017near,lee2016blind,jung2018blind,Chen2021,Ahmed2013blind,Ling_2015}
and references therein). We may in particular take additional structures of
$\beta_{A}$ and $\beta_{B}$ into account, for example if their values are
drawn from a finite, known alphabet. The HiHTP-algorithm here rather fulfills
a theoretical role -- since it \emph{provably} recovers every signal as soon
as the HiRIP is given, and the HiRIP can be guaranteed with high probability
using a reasonable codebook $Q$.

\subsection{Security}

In the last section, we saw that for a reasonable $Q$, the blind deconvolution
problem $y=C(h\otimes\beta)+n$ can be solved for $h\otimes\beta$. This of
course also applies to Eve -- she can (assuming no noise) from her
measurements gain access to $h_{A\rightarrow E}\otimes\beta_{A}+\gamma
h_{B\rightarrow E}\beta_{B}$. Note however that this datum only partially
reveals $h_{AB}$ and $\beta_{A}$, $\beta_{B}$. If the filters are assumed
noise-free (i.e. $\varsigma^{2}=0$), it only reveals the superposition of
$h_{AB}\otimes\beta_{A}$ and $h_{AB}\otimes\beta_{B}$, and cannot directly
infer individual contributions of Alice's and Bob's sequences. Concretely, she
must in this case determine the codeword only from the knowledge of
$h_{AB}\otimes(\beta_{A}+\gamma\beta_{B})$. The purpose of the first part of
this section section is to prove that if $\gamma\approx1$, \emph{she
essentially can not}.

In the second part of the section, we will also tackle the more realistic case
of $\varsigma, \sigma\neq0$. We will show that as long as $\sigma$ is
comparable to $\varsigma$, and the sparsity of the filter $h$ is smalled
compared to $k$, the codeword is almost as secure as before. If $s$ is large,
we can still guarantee security if $\varsigma$ is considerably smaller than
$\sigma$.

Formally, we will give an estimate of the \emph{conditional entropies} of the
codeword $\mathfrak{c}$ given the knowledge $h_{A\to E} \otimes(\beta
_{A}+\gamma h_{B\to E}\beta_{B})$. Let's repeat its definition. Given two
discrete random variables $X$ and $Y$, the conditional entropy of $X$ given
$Y$ is equal to
\[
H(X\,|\,Y)=-\sum_{\substack{x\in\calX\\y\in\calY}}\mathbb{P}(X=x,Y=y)\log
\left(  \prb{ X=x \, \vert \, Y=y}\right)  .
\]
where the sum ranges over all possible values $\calX$ of $x$ and $\calY$ of
$y$, respectively. For a continuous variable, an integral is instead used:
\begin{align*}
H(X \, \vert\, Y) = -\int p(x, y)\log\left(  p(x\, \vert\,y)\right)
\mathrm{d}x \mathrm{d}y.
\end{align*}
Let us first record an estimate the amount of entropy contained in $\beta
_{A}+\gamma\beta_{B}$.

\begin{theo}
\label{theo:underlying_entropy} Suppose that the $\alpha_{A},\alpha_{B}%
\in\lbrack0,1]^{k}$ are uniformly drawn from the real interval $[0,1]$. Define%
\[
H_{\gamma}(k):=-\log\left(  \binom{2k}{k}^{-1}\frac{1-\delta^{2k}}%
{(1-\delta)^{k}}+\delta^{k}\right)
\]
where $\gamma:=1-\delta$. Then, conditioned on the event that $\beta_{A}$ and
$\beta_{B}$ have disjoint supports,%
\[
H(\beta_{A}\,|\,\beta_{A}+\gamma\beta_{B})\geq H_{\gamma}(k).
\]

\end{theo}

The proof of the above result is just technical, and does not provide much
insight. In the interest of brevity, we give it in Appendix
\ref{app:deltaproof}.

\begin{rem}
\label{rem:stirling} The attentive reader may already have noticed that the
above theorem gives a vacuous statement when $\gamma\to0$. This should however
not bother us -- it is clear that all security is lost when $\gamma$ is small,
since it will be possible to distinguish with high probability which
coefficients in $\beta_{A}+\gamma\beta_{B}$ belong to Alice simply by looking
at their magnitude.

The interesting case is instead when $\gamma$ is close to $1$. Note that by a
Taylor expansion, and Sterling's approximation $\log(\binom{2k}{k}%
)~\approx~2k$, we have
\begin{align*}
H_{\gamma}(k) \approx\log\left(  \binom{2k}{k}\right)  - k \delta
\approx2k-k\delta= k(1+\gamma)
\end{align*}
for small values of $\delta$.

\end{rem}

\subsubsection{\noindent\textbf{\underline{The noiseless case}}}

\hfill


\begin{ass}
As advertised above, let us now first tackle the noiseless case, i.e.
$\varsigma=\sigma=0$. Under these assumptions, Eve receives%
\[
y_{E}=h_{AB}\ast Q(\beta_{A}+\gamma\beta_{B})+n_{E}.
\]
Here, Eve gains knowledge of $h_{AB}\otimes(\beta_{A}+\gamma\beta_{B})$ by
solving her blind deconvolution problem. The entity to estimate is therefore
$H(\mathfrak{c}\,|\,h_{AB}\otimes(\beta_{A}+\gamma\beta_{B})$.
\end{ass}

The formal result is as follows.

\begin{theo}
\label{th:security} {For almost every value of $h$},
\[
H(\mathfrak{c}\,|\,h_{AB}\otimes(\beta_{A}+\gamma\beta_{B}))\geq\left(
1-\tfrac{17k^{4}}{n}\right)  \cdot(H_{\gamma}(k)-1).
\]

\end{theo}

Considering Remark \ref{rem:stirling}, the message of the theorem is as
follows: For large values of $k$ and $n$ and $\gamma\approx1$, the conditional
entropy of the codeword given $h_{AB}\otimes(\beta_{A}+\gamma\beta_{B})$ is,
up to small approximation error, larger than $k(1+\gamma)-1$ nats. Hence, Eve
does not have an easier time to determine the codeword than it is for her to
determine which of the peaks belong to Alice and which belong to Bob in the
toy example in the beginning of the article.

The proof of Theorem \ref{th:security} proceeds in several steps. The first
step of the argumentation is to eliminate the channel $h_{AB}$.

\begin{lem}
\label{lem:reduction} For almost every draw of $h$,
\[
H(\widehat{h_{AB}^{\Uparrow}}\cdot\widehat{\beta_{A}^{\uparrow}}\cdot
\widehat{\beta_{B}^{\uparrow}}\,|\,h\otimes(\beta_{A}+\gamma\beta
_{B}))=H(\beta_{A}\ast\beta_{B}\,|\,\beta_{A}+\gamma\beta_{B}).
\]

\end{lem}

\begin{proof}
This lemma is more or less in its entirety a consequence of the following fact
(which is a direct consequence of the definition of conditioned entropy): If
$X$ is a random variable on $\calX$ and $Y$ a random variable on $\calY$, and
$f_{X}:\calX \to\calX^{\prime}$, $f_{Y}: \calY \to\calY^{\prime}$ are
injective functions,
\begin{align*}
H(f_{Y}(Y) \, \vert\, f_{X}(X)) = H(Y \, \vert\, X).
\end{align*}

It is clear that for $h\neq0$ fixed, i.e. almost surely, the map $\beta\mapsto
h\otimes\beta$ is injective. If each entry of $\widehat{h^{\Uparrow}}$ is
non-zero, the map $\widehat{\beta}\mapsto\widehat{h^{\Uparrow}}\cdot
\widehat{\beta}$ is also injective. Let us argue that that almost surely is
the case. Each entry of $\widehat{h^{\Uparrow}}$ is a linear combination of
the $\{\gamma_{k}\}_{k\in\lbrack s]}$. Conditioned on the draw of the support
of $h$, it is hence, as a sum of independent non-degenerate Gaussians, itself
a Gaussian. Hence, it almost never vanishes. By the law of total probability,
the same is true also not conditioned on the draw of $\mathrm{supp}h$.

Thus, almost surely with respect to the draw of $h_{AB}$, we have
\begin{align*}
&  H(\widehat{h_{AB}^{\Uparrow}}\cdot\widehat{\beta_{A}^{\uparrow}}%
\cdot\widehat{\beta_{B}^{\uparrow}}\,|\,h_{AB}\otimes(\beta_{A}+\gamma
\beta_{B}))\\
&  \quad=H(\widehat{\beta_{A}^{\uparrow}}\cdot\widehat{\beta_{B}^{\uparrow}%
}\,|\,\beta_{A}+\gamma\beta_{B})=H(\beta_{A}^{\uparrow}\ast\beta_{B}%
^{\uparrow}\,|\,\beta_{A}+\gamma\beta_{B}),
\end{align*}
where the final step follows from injectivity of the Fourier transform. Since
the above is true almost surely, it is also true unconditioned on the draw of
$h_{AB}$.

To finish the proof, note that since the map $\beta_{A}\ast\beta_{B}%
\mapsto(\beta_{A}\ast\beta_{B})^{\uparrow}=\beta_{A}^{\uparrow}\ast\beta
_{B}^{\uparrow}$ is injective, we get
\[
H(\beta_{A}^{\uparrow}\ast\beta_{B}^{\uparrow}\,|\,\beta_{A}+\gamma\beta
_{B})=H(\beta_{A}\ast\beta_{B}\,|\,\beta_{A}+\gamma\beta_{B}),
\]
which is the claim.
\end{proof}

We now move on to the actual bulk of the proof: to analyze the conditional
entropy of $\beta_{A}\ast\beta_{B}$ given $\beta_{A}+\gamma\beta_{B}$. We
begin with a crucial lemma.


\begin{lem}
\label{lem:inj} Let $\Sigma$ be a subset of $[n]$ of cardinality $2k$, and
$\nu= \sum_{c\in\Sigma} \alpha_{c}e_{c}$ a vector in $\C^{n}$, with
$\alpha_{c} \neq0$ for each $c\in\Sigma$. Suppose that the set $\Sigma+
\Sigma= \set{ a+b \text{ mod } n \, \vert \, a,b \in \Sigma }$ of sums between
elements in $\Sigma$ has cardinality $k(2k-1)$. Define $M$ as the set of
vectors of the form $\mu=\sum_{a \in\sigma} \alpha_{a} e_{a}$, where $\sigma$
ranges over the subsets of $\Sigma$ of cardinality $k$. Then,
\begin{align}
\Psi: M \to\C^{n}, \mu\mapsto\mu*(\nu-\mu)
\end{align}
is injective modulo the equivalence $\mu\sim\nu-\mu$.
\end{lem}


\begin{proof}
We want to prove that if $\mu\not \sim \mu^{\prime}$, $\Psi(\mu)\neq\Psi
(\mu^{\prime})$. Note that the latter is the case if and only if there exists
a $\phi\in\C^{n}$ with $\sprod{\phi,\Psi(\mu)}\neq
\sprod{\phi,\Psi(\mu^{\prime})}$. We have
\begin{align*}
\sprod{\phi,\Psi(\mu)}  &  = \sprod{\phi,\mu*(\nu-\mu)} = \sum_{a\in[n]}
\overline{\phi_{a}} (\mu*(\nu-\mu))_{a}\\
&  = \sum_{a,b \in[n]} \overline{\phi_{a}} \mu_{\ell}(\nu-\mu)_{a-b} =
\sum_{a,b \in[n]} \overline{\phi_{a}} \mu_{b}(\nu-\mu)_{a}\\
&  = \sum_{a \in\sigma, b \in\Sigma\backslash\sigma} \overline{\phi_{a+b}}
\alpha_{a}\alpha_{b},
\end{align*}
and similarly for $\mu^{\prime}$. Now, let us assume we can prove the
existence of an element $c_{0} = a_{0}+b_{0}$ with
\begin{align}
\label{eq:magical_element}c_{0} \in\sigma^{\prime}+ (\Sigma\backslash
\sigma^{\prime})\text{ or } c_{0} \in\sigma+ (\Sigma\backslash\sigma), \text{
\emph{but not both}.}%
\end{align}
Then, choosing $\phi= e_{c_{0}}$, all terms in the above sum would be zero, so
that $\sprod{\phi,\Psi(\mu)} = 0$. However, $\sprod{\phi, \Psi(\mu')} =
\alpha_{a_{0}}\alpha_{b_{0}} \neq0$. Hence, the claim would follow.

So let us construct such an element. Since $\mu\not \sim \mu^{\prime}$, the
support $\sigma$ of $\mu$ is neither equal to the support $\sigma^{\prime}$ of
$\mu^{\prime}$, nor to $\Sigma\backslash\sigma^{\prime}$. We distinguish two cases.

\underline{Case I: $\sigma$ and $\sigma^{\prime}$ are not disjoint}. In this
case, there exists an $a_{0}\in\sigma\cap\sigma^{\prime}$. There must however
also exist a $b_{0}$ which is either in $\sigma\backslash\sigma^{\prime}$ or
$\sigma^{\prime}\backslash\sigma$ -- otherwise, $\sigma=\sigma^{\prime}$,
which we ruled out. Due to symmetry reasons, we may without loss of generality
assume that the former is the case, i.e. that there exists
\[
a_{0}\in\sigma\cap\sigma^{\prime},b_{0}\in\sigma\cap(\Sigma\backslash
\sigma^{\prime}).
\]
That means that $(a_{0},b_{0})\in\sigma^{\prime}\times(\Sigma\backslash
\sigma^{\prime})$, but that neither $(a_{0},b_{0})$ nor $(b_{0},a_{0})$ is in
$\sigma\times(\Sigma\backslash\sigma)$. Now, $\Sigma+\Sigma$ having $k(2k-1)$
elements, means that all sums $a+b$, $a,b\in\Sigma$ are distinct (excluding
that $a+b=b+a$). In other words, the operation $(a,b)\rightarrow(a+b)$ is
injective on $\Sigma\times\Sigma/\sim$, where $\sim$ identifies $(a,b)$ and
$(b,a)$. The latter together with $(a_{0},b_{0}),(b_{0},a_{0})\notin
\sigma\times(\Sigma\backslash\sigma)$ implies \eqref{eq:magical_element}.

\underline{Case II: $\sigma$ and $\sigma^{\prime}$ are disjoint.} Either
$\sigma$ or $\sigma^{\prime}$ must be nonempty, without loss of generality
$\sigma$. Then, due to the disjointness, $\emptyset\neq\sigma\sse\Sigma
\backslash\sigma^{\prime}$. However, it can't be $\Sigma\backslash
\sigma^{\prime}\sse\sigma$, since then, $\sigma=\Sigma\backslash\sigma
^{\prime}$. Hence, there exists $a_{0}$ and $b_{0}$ with
\[
a_{0}\in\sigma\cap(\Sigma\backslash\sigma^{\prime}),\quad b_{0}\in
(\Sigma\backslash\sigma)\cap(\Sigma\backslash\sigma^{\prime})
\]
That is, $(a_{0},b_{0})\in\sigma\times(\Sigma\backslash\sigma)$ but
$(a_{0},b_{0}),(b_{0},a_{0})\notin\sigma^{\prime}\times(\Sigma\backslash
\sigma^{\prime})$. Using the same argument as above, we obtain that
$a_{0}+b_{0}$ is an element with \eqref{eq:magical_element}.
\end{proof}

The idea of the proof of our main result will now be to apply the above result
with $\Sigma$ equal to the combined support of $\beta_{A}$ and $\beta_{B}$,
$\Sigma_{AB}=\supp\beta_{A}\cup\supp\beta_{B}$. In order to do so,
$\Sigma_{AB}$ must fulfill the assumption of the lemma. Therefore, we prove
the following simple bound on that being the case.

\begin{lem}
\label{lem:ineq} Let $\Sigma_{AB}=\supp\beta_{A}\cup\supp\beta_{B}$ and $E$
denote the event
\[
E=\set{ \abs{\Sigma_{AB}}=2k, \quad \abs{\Sigma_{AB} +\Sigma_{AB}}= k(2k-1).}
\]
Then
\begin{equation}
\prb{E^c}\leq\tfrac{17k^{4}}{n} \label{eq:bound}%
\end{equation}

\end{lem}

\begin{proof}

Let us fix $\alpha_{A},\alpha_{B}$ (so that we have countable probability
space). To ensure that $\Sigma_{AB}$ has $2k$ entries, we need $x^{\prime}\neq
y^{\prime}\mod n$ for each $x^{\prime}\in\supp\beta_{A}$ and $y^{\prime}%
\in\supp\beta_{B}$. Each such event has a probability of $\frac{1}{n}$, and
there are less than $k^{2}$ of them. Hence, the total probability of the union
of them is not larger than $\sfrac{k^2}{n}\leq\sfrac{k^4}{n}$.

We bound the probability of $\Sigma_{AB}+\Sigma_{AB}$ not having $k(2k+1)$
elements in the same manner: Here, we get less than $(2k)^{4}=16k^{4}$ events
$a^{\prime}+b^{\prime}\neq c^{\prime}+d^{\prime}$. for different values of
$a,b,c,d$, and the probability of such each event is smaller than
$\sfrac{1}{n}$, leading to a total probability smaller than $\sfrac{16k^4}{n}$%
.
%
Combining the two bounds and averaging over $\alpha_{A},\alpha_{B}$ yields the claim.
\end{proof}

Now, we may prove Theorem \ref{th:security}.

\begin{proof}
[Proof of Theorem \ref{th:security}]Again fix $\alpha_{A},\alpha_{B}$. First,
by Lemma \ref{lem:reduction}, it is enough to argue that
\[
H(\beta_{A}\ast\beta_{B}\,|\,\beta_{A}+\gamma\beta_{B})\geq H_{\gamma
}(k)\left(  1-\frac{17k^{4}}{n}\right)  .
\]
To simplify the notation, let us write $\beta_{A}\ast(\gamma\beta_{B})=C$ and
$\beta_{A}+\gamma\beta_{B}=S$. It is clear that we are trying to bound
$H(C\,|S)$. Remembering the definition of the event $E$ from Lemma
\ref{lem:ineq}, we may split the sum defining $H(C|S)$ as follows
\begin{align*}
H(C|S)=  &  -\sum_{(c,s)\in E}\prb{C=c,S=s}\log(\prb{C=c \, \vert \, S=s})\\
&  -\sum_{(c,s)\notin E}\prb{C=c,S=s}\log(\prb{C=c \, \vert \, S=s})
\end{align*}
As for the sum over the complement of $E$, we can trivially declare a lower
bound of zero.

As for the other sum, let us fix $s$.
On the event $E$, $\Sigma=\supp S$ has $2k$ elements, and hence, $\beta_{A}$
and $\beta_{B}$ must have disjoint support. Consequently, $C$ is given by
$\mu\ast(S-\mu)$, where $\mu$ is distributed on the set $\{\sum_{x\in\sigma
}\alpha_{x}e_{x},\sigma\sse\Sigma,\abs{\sigma}=k\}$. Lemma \ref{lem:inj} now
states that on the event $E$, the map $\Psi:\mu\rightarrow\mu\ast(s-\mu)$ is
injective modulo the equivalence $\mu\sim(s-\mu)$. This means, for every value
of $c$, there exists exactly two values $a,a^{\prime}$ so that
\[
\mathbb{P}(C=c,S=s)=\mathbb{P}(\beta_{A}=a,S=s)+\mathbb{P}(\beta_{A}%
=a^{\prime},S=s)).
\]
Using the concavity of $x\mapsto-x\log(x)$, we can therefore bound
\begin{align*}
&  \mathbb{P}(C=c,S=s)\log(\mathbb{P}(C=c,S=s))\\
&  \qquad\geq-\mathbb{P}(\beta_{A}=a,S=s)\log(2\mathbb{P}(\beta_{A}=a,S=s)))\\
&  \quad\qquad-\mathbb{P}(\beta_{A}=a^{\prime},S=s)\log(2\mathbb{P}(\beta
_{A}=a^{\prime},S=s)),
\end{align*}
and by utilizing the definition of conditional probability
\begin{align*}
&  -\mathbb{P}(C=c,S=s)\log(\mathbb{P}(C=c,S=s))\\
&  \qquad\geq-\mathbb{P}(\beta_{A}=a,S=s)(\log(\mathbb{P}(\beta_{A}%
=a|S=s)))+1)\\
&  \quad\qquad+-\mathbb{P}(\beta_{A}=a^{\prime},S=s)(\log(\mathbb{P}(\beta
_{A}=a^{\prime}|S=s))+1)
\end{align*}
Summing over the values of $c$ on the right-hand side amounts to summing over
all values of $a$ on the left hand side. Hence, we get that
conditioned on $E$, the conditional entropy of $\beta_{A}\ast(s-\beta_{A})$
given $\beta_{A}+\gamma\beta_{B}$ bounded below by the conditional entropy of
$\beta_{A}$ given $\beta_{A}+\gamma\beta_{B}$, minus
\[
\sum_{a,s\in E}\mathbb{P}(\beta_{A}=a,S=s)=\mathbb{P}(E)
\]
By Theorem \ref{theo:underlying_entropy}, that conditional probability is
bounded below by $H_{\gamma}(k)$. Therefore,
\begin{align*}
&  -\sum_{(c,s)\in E}\prb{C=c,S=s}\log(\prb{C=c \, \vert \, S=s})\\
&  \geq\prb{E}(H_{\gamma}(k)-1)
\end{align*}
Now it is only left to utilize our bound on $\mathbb{P}(E^{c})$ from Lemma
\ref{lem:ineq}. Again averaging over $\alpha_{A},\alpha_{B}$ yields the final result.
\end{proof}


In order for the bound \eqref{eq:bound} to be relevant, we need $n$ to be of
the order $k^{4}$. We conjecture that this is overly pessimistic, since Lemma
\ref{lem:inj} by no means exactly characterizes when the map $\Psi$ is
injective. Note however that $n$ only appears in a logarithmic term in the
required sample complexity $\mu$ for the HiHTP algorithm to be successful.
Hence, even when $n$ is of order $k^{4}$, the HiHTP-algorithm can be used for
blindly deconvolving $h\ast Q\beta$, with $\mu/\log{\mu}\sim sk\log(n)$.

\subsubsection{\textbf{\underline{The noisy case}}}

\hfill

\label{sec:noisy} We move on to analysing the noisy case, i.e. when
$\varsigma, \sigma\neq0$. We still assume that Eve is able to solve her blind
deconvolution problem in the following sense.

\begin{ass}
\label{ass:noisysolve} Eve can from the noisy measurements $y_{E}$ determine
the correct support of
\[
h_{A\to E} \otimes\beta_{A} + h_{B\to E} \otimes\beta_{B},
\]
and the operator $C$ restricted to tensors with that support is well
conditioned: $\norm{C_{\supp T}^*C_{\supp T}- I_{\supp T}} \leq\delta$ for
some $\delta<1$.

Subsequently, she determines $T$ as $C_{\supp T}^{-1}y_{E}$. Consequently, she
gains access to
\[
T:=h_{A\rightarrow E}\otimes\beta_{A}+h_{B\rightarrow E}\otimes\beta
_{B}+\underline{n},
\]
where $\overline{n}$ is a Gaussian with mean zero and covariance
$\overline{\Sigma}:=\sigma^{2}C_{\supp T}^{-1}C_{\supp T}^{-\ast}$. 

\end{ass}

\begin{rem}
    One should note that the above recovery strategy is slightly suboptimal: Eve could in addition to the knowledge of $\supp T$ also choose to take the  low rank structure of the signal $h_{A\to E}\otimes \beta_A + h_{B\to E}\otimes \beta_B$. However, since Eve is assumed to have access to $\supp T$, the only effect this would have is to change the distribution of the anyhow small noise vector $\underline{n}$ slightly.
\end{rem}
Assuming the above attacker model, we will now proceed as in the last section
and estimate $H(\mathfrak{c} \, \vert\, h_{A\to E} \otimes\beta_{A} + h_{B\to
E} \otimes\beta_{B} + \underline{n})$. The formal result will be as follows.

\begin{theo}
\label{th:security_2} Define
\[
H_{\sigma,\varsigma}(s)=s\ln(1+2k\tfrac{\varsigma^{2}}{\sigma^{2}}).
\]
Then, for almost every value of $h$
\begin{align*}
&  H(\mathfrak{c}\,|\,h_{A\rightarrow E}\otimes\beta_{A}+h_{B\rightarrow
E}\otimes\beta_{B}+\underline{n})\geqsim\\
&  \quad\left(  1-\tfrac{17k^{4}}{n}\right)  (H_{\gamma}(k)-1)-H_{\sigma
,\varsigma}(s)
\end{align*}

\end{theo}

Before we prove the claim some comments are in order again.

\begin{rem}
In the low noise regime, $\tfrac{\varsigma^{2}}{\sigma^{2}}k\ll1$, we can
approximate
\[
H_{\sigma,\varsigma}(s)\approx2ks\tfrac{\varsigma^{2}}{\sigma^{2}}.
\]
In the regime $\gamma\approx1$ and $n\geqsim k^{4}$, the lower bound on the
relative entropy can hence be bounded below (up to constants) by
\[
k(1+\gamma-2s\tfrac{\varsigma^{2}}{\sigma^{2}})
\]
We see that if $\varsigma^{2}\geq\theta s\sigma^{2}$ for some $\theta<1$, the
relative entropy of the codeword is still proportional to $k$. Hence, if the
devations $n_{A}$, $n_{B}$ are about a factor $s^{\sfrac{1}{2}}$ smaller then
the measurement noise, they will 'drown', and the codeword is kept secret.

That the filter deviations need to be small is not really restrictive
since it is only in this regime where Eve has a chance to get a really good
estimate for $h$, which she needs to recover the codeword.
\end{rem}

\begin{proof}
Using the same argument as in Lemma \ref{lem:reduction}, we see that it is
enough to bound
\[
H(\beta_{A}\ast\beta_{B}\,|\,h_{A\rightarrow E}\otimes\beta_{A}%
+h_{B\rightarrow E}\otimes\beta_{B}+\underline{n})
\]
from below. Using, in addition to the notation $T$, the notations $C=\beta
_{A}\ast\beta_{B}$ and $S=\beta_{A}+\gamma\beta_{B}$, we see that we are
trying to bound $H(C\,|,T)$. Let us begin with the following trivial bound
\[
H(C\,|\,T)\geq H(C\,|T,S).
\]
This bound is intuitively quite sharp for small values of $\sigma$ and
$\varsigma$, since $T$ essentially is a noisy version of $S$. The Bayes rule,
conditioned on $S$, now gives
\[
H(C\,|\,T,S)=H(T\,|\,C,S)-H(T\,|\,S)+H(C\,|\,S).
\]
We now analyse each of these terms individually. First, by Lemma
\ref{lem:reduction}, we almost surely over the draw of $h$ have
\begin{align*}
H(C\,|\,S)  &  =H(\beta_{A}\ast\beta_{B}\,|\,h\otimes(\beta_{A}+\gamma
\beta_{B}))\\
&  =H(\beta_{A}\ast\beta_{B}\,|\,\beta_{A}+\gamma\beta_{B})\\
&  \geq\left(  1-\tfrac{17k^{4}}{n}\right)  (H_{\gamma}(k)-1),
\end{align*}
where we in the final step in particular applied Theorem \ref{th:security}.

Now let us examine $H(T\,|\,C,S)$. We can surely bound
\[
H(T\,|\,C,S)\geq H(T\,|\,C,S,\beta_{A},\beta_{B})=H(T\,|\,\beta_{A},\beta
_{B}).
\]
Conditioned on the values for $\beta_{A}$ and $\beta_{B}$, $T$ is a Gaussian,
with mean $h\otimes(\beta_{A}+\gamma\beta_{B})$ and covariance
\[
\Sigma_{\beta_{A},\beta_{B}}:=\overline{\Sigma}+\varsigma^{2}\id_{\supp
h}\otimes(\beta_{A}\beta_{A}^{\ast}+\beta_{B}\beta_{B}^{\ast}).
\]
We can use this fact to bound $H(T\,|\,S)$. If we denote the PDF of a Gaussian
with covariance $\Sigma$ and mean $h_{AB}\otimes(\beta_{A}+\gamma\beta_{B})$
with $\psi_{\Sigma_{\beta_{A},\beta_{B}}}$, the density of $T$ conditioned on
$S$ will be equal to
\[
\varphi(v)=\mathbb{E}^{\prime}(\varphi_{\Sigma_{\beta_{A},\beta_{B}}}(v)),
\]
where the expected value is over the draw of (the support) of $\beta_{A}$. Due
to the concavity of the function $x\mapsto-x\log(x)$, we may now argue that
\[
-\log(\varphi(v))\varphi(v)\leq\mathbb{E}^{\prime}\left(  -\log(\varphi
_{\Sigma_{\beta_{A},\beta_{B}}}(v))\varphi_{\Sigma_{\beta_{A},\beta_{B}}%
}(v)\right)  ,
\]
and from that
\[
H(T\,|\,S)\leq\mathbb{E}^{\prime}(H(T\,|\,S,\beta_{A},\beta_{B}))=\mathbb{E}%
^{\prime}(H(T\,|\,\beta_{A},\beta_{B}))
\]
Now, a Gaussian with covariance $\Sigma$ has the entropy
\[
\tfrac{1}{2}\log(\det(2\pi e\Sigma)).
\]
Consequently,
\begin{align*}
&  H(T\,|\,C,S)-H(T\,|\,S)\geq\\
&  \quad\tfrac{1}{2}\left(  \log(\det(2\pi e\Sigma_{\beta_{A},\beta_{B}%
}))-\mathbb{E}^{\prime}\log(\det(2\pi e\Sigma_{\beta_{A},\beta_{B}}))\right)
\end{align*}
An elementary calculation shows that $\Sigma_{\beta_{A},\beta_{B}}$ equals
\[
\overline{\Sigma}(\id_{\supp T}+\varsigma^{2}\overline{\Sigma}^{-\sfrac{1}{2}}%
\left(  \id_{\supp h_{AB}}\otimes(\beta_{A}\beta_{A}^{\ast}+\beta_{B}\beta
_{B}^{\ast})\right)  \overline{\Sigma}^{-\sfrac{1}{2}}).
\]
Let us estimate
\begin{equation}
\det(\id_{\supp T}+\varsigma^{2}\overline{\Sigma}^{-\sfrac{1}{2}}\left(
\id_{\supp h_{AB}}\otimes(\beta_{A}\beta_{A}^{\ast}+\beta_{B}\beta_{B}^{\ast
})\right)  \overline{\Sigma}^{-\sfrac{1}{2}})\label{eq:det}%
\end{equation}
which we do so via directly estimating the eigenvalues of the matrix. For
that, it is clearly enough to control the ones of $\varsigma^{2}%
\overline{\Sigma}^{-\sfrac{1}{2}}\left(  \id_{\supp h_{AB}}\otimes(\beta
_{A}\beta_{A}^{\ast}+\beta_{B}\beta_{B}^{\ast})\right)  \overline{\Sigma
}^{-\sfrac{1}{2}}$ First, note that when $\overline{\Sigma}=\id$, the
eigenvalues are easy to write down: since $\beta_{A}$ and $\beta_{B}$ have
disjoint supports, $s$ of the eigenvalues are equal to $\varsigma
^{2}\norm{\beta_A}^{2}$, $s$ are equal to $\varsigma\norm{\beta_B}^{2}$, and
the remainding $s(k-2)$ are equal to $0$. Now, $\overline{\Sigma
}^{-\sfrac{1}{2}}$ is, by assumption \ref{ass:noisysolve}, a matrix with
spectrum in $[\sigma^{-2}(1-\delta),\sigma^{-2}(1+\delta)]$. The quantitative
version of Sylvester's law of inertia \cite{ostrowski1959quantitative} now
shows that when $\overline{\Sigma}\neq0$, the first and second sets of
eigenvalues lie in between $\tfrac{\varsigma^{2}}{\sigma^{2}}%
\norm{\beta_A}^{2}(1-\delta)$ and $\tfrac{\varsigma^{2}}{\sigma^{2}%
}\norm{\beta_A}^{2}(1-\delta)$, the second between $\tfrac{\varsigma^{2}%
}{\sigma^{2}}\norm{\beta_B}(1-\delta)$ $\tfrac{\varsigma^{2}}{\sigma^{2}%
}\norm{\beta_B}(1+\delta)$, and the rest are still equal to zero. We conclude
that the determinant \eqref{eq:det} is bigger than
\[
(1+(1-\delta)\tfrac{\varsigma^{2}}{\sigma^{2}}\norm{\beta_A}^{2}%
)^{s}(1+(1-\delta)\tfrac{\varsigma^{2}}{\sigma^{2}}\norm{\beta_B}^{2})^{s}%
\geq1
\]
and smaller than
\begin{align*}
&  (1+(1+\delta)\tfrac{\varsigma^{2}}{\sigma^{2}}\norm{\beta_A}^{2}%
)^{s}(1+(1+\delta)\tfrac{\varsigma^{2}}{\sigma^{2}}\norm{\beta_B}^{2})^{s}\\
&  \qquad\leq(1+k\tfrac{\varsigma^{2}}{\sigma^{2}}(1+\delta))^{2s},
\end{align*}
where we applied the trivial bound $\norm{\beta_A}^{2},\norm{\beta_B}^{2}\leq
k$ in the final step. This shows that
\begin{align*}
&  \tfrac{1}{2}\left(  \log(\det(2\pi e\Sigma_{\beta_{A},\beta_{B}%
}))-\mathbb{E}^{\prime}\log(\det(2\pi e\Sigma_{\beta_{A},\beta_{B}}))\right)
\geq\\
&  \ \tfrac{1}{2}\big(\log(\det(2\pi e)\overline{\Sigma})+\log(1)\\
&  \ \ -\log(\det(2\pi e\overline{\Sigma}))-\log((1+k\tfrac{\varsigma^{2}%
}{\sigma^{2}}(1+\delta))^{2s})\big)=-H_{\sigma,\varsigma}(s),
\end{align*}
where we in particular estimated $(1+\delta)\leq2$. The claim has been proven.
\end{proof}

We can now summarize the result in the following theorem.

\begin{theo}
Give $m$ rounds of the above procedure and employing a universal hash function
taken for the collective quantized common secrets, a uniformly destributed
secret key with rate per round%
\[
k\geq\beta H_{\gamma}(k)-H_{\sigma,\varsigma}(s)
\]
for some $0<\beta<1$ is achievable provided $m,n$ and $\mu$ are sufficiently large.
\end{theo}

\begin{proof}
The proof is pretty standard. For some $k$, fix the signal space dimension $n,\mu$ such that any
probability of error in the reconciled key is small say below some
$\epsilon>0$ guided by equations \ref{eq:mubound} and \ref{eq:bound} with only logarithmic efficiency
penalty in $n$. By Theorem \ref{th:security_2}, for a single round the entropy of the secret
$\mathfrak{c}$ is given by%
\[
H(\mathfrak{c}\,|\,h_{A\rightarrow E}\otimes\beta_{A}+h_{B\rightarrow
E}\otimes\beta_{B}+\underline{n})\geq\beta H_{\gamma}(k)-H_{\sigma,\varsigma
}(s).
\]
for any $\beta<1$ such that $\beta H_{\gamma}(k)-H_{\sigma,\varsigma}(s)>0$.
After $m$ rounds we collect the secret $\mathfrak{c}^{(m)}=\mathfrak{c}%
_{1},...,\mathfrak{c}_{m}$ and it is well-known that the collision and
min-entropy is at least
\begin{align*}
& H_{c}(\mathfrak{c}^{(m)}\,|\,h_{A\rightarrow E}\otimes\beta_{A}%
+h_{B\rightarrow E}\otimes\beta_{B}+\underline{n})\\
& \geq H_{\infty}(\mathfrak{c}^{(m)}\,|\,h_{A\rightarrow E}\otimes\beta
_{A}+h_{B\rightarrow E}\otimes\beta_{B}+\underline{n})\\
& \geq m\left(  \beta H_{\gamma}(k)-H_{\sigma,\varsigma}(s)-c_{1}\left(
m\right)  -c_{2}\left(  \epsilon\right)  \right)  =:k^{\prime},
\end{align*}
by virtue of the asymtotic equipartition property \cite[Lemma 4.4]{pls_book}.
Here $\lim_{m\rightarrow\infty}c_{1}\left(  m\right)  =0$ and $\lim
_{\epsilon\rightarrow\infty}c_{2}\left(  \epsilon\right)  =0$. We leave out
the technical subtlety that Eve's observations are uncountable.

Then we can
apply any (vector) quantizer $\mathcal{Q}$ to common secrets $\mathfrak{c}%
^{(m)},\mathfrak{c}^{\prime(m)}$ so that $\mathcal{Q}:C^{mn\mu}\rightarrow
\{0,1\}^{\theta mn\mu},\mathfrak{c}^{(m)},\mathfrak{c}^{\prime(m)}%
\hookrightarrow\mathfrak{c}_{b},\mathfrak{c}_{b}^{\prime}$, for some natural
$\theta\geq1$. The quantizer is required to be one-to-one so that the
collision entropy or min-entropy is the same. Finally applying a a universal
hash function $\mathcal{G}:\{0,1\}^{\theta mn\mu}\rightarrow\{0,1\}^{k^{\prime
}},\mathfrak{c}_{b},\mathfrak{c}_{b}^{\prime}\hookrightarrow\mathfrak{k}%
_{b},\mathfrak{k}_{b}^{\prime}$ where secret keys $\mathfrak{k}_{b}%
,\mathfrak{k}_{b}^{\prime}$ are almost uniformly distributed over the
$2^{k^{\prime}}$ possibilities \cite[Theorem 4.4]{pls_book} and $\mathfrak{k}%
_{b}=\mathfrak{k}_{b}^{\prime}$ with probability below $\epsilon$ for
$m,n,\mu$ large enough.
\end{proof}

\begin{rem}
Notably, before closing this section, we would like to emphasize that our
attacker model is in many ways rather beneficial for Eve:
\begin{itemize}
\item We assume that she can perfectly recover the support of $h_{AB}%
\otimes(\beta_{A}+\gamma\beta_{B})$, which is an object of higher effective
dimension than the individual contributions $h_{AB}\otimes\beta_{A}$ and
$h_{AB}\otimes\beta_{B}$ that Bob and Alice, respectively, need to recover.
Note however that the dimension of Eve's measurement is the same Alice's and
Bob's, respectively. Hence, if Alice and Bob operate in the 'phase transition
regime' for them \cite{jung2018blind}, which could be monitored, any key
recovery of Eve attempt must necessarily fail. However the phase transition
regime is hard to quantify. We have incorporated some simulations related to
this in Section \ref{sec:experiments}.

\item As was mentioned, the assumption that the filter(s) $h_{A\rightarrow
E},h_{B\rightarrow E}$ is equal (are close) to $h_{AB}$ completely removes the
need for Eve to estimate $h_{AB}$ once she has solved her blind deconvolution
problem. In general, both $h_{A\rightarrow E},h_{B\rightarrow E}$ can be
only weakly correlated to $h_{AB}$, which would intuitively make it hard for Eve to
estimate $\mathfrak{c}$.

\item Finally, Eve is given access to the codebook $Q$, which Alice and Bob
could have kept secret. This is an interesting line of future work.
\end{itemize}

All of these other possible layers of security are complicated to quantify
theoretically, and we deem an analysis of them beyond the scope of this work.
This includes the potential use of massive MIMO \cite{Wunder2019_TCOM}.
\end{rem}

Let us now step to the experiment section.

\section{Experiments}

\label{sec:experiments}

In order to practically verify our assumptions about the shared secret that
Alice and Bob compute on behalf of the output of the HiHTP algorithm, we
conducted several experiments with different parameters to simulate a
real-life execution of our key generation protocol.

\subsection{Signal and channel sparsity}

As the relation between sparsity and vector dimension is the decisive
parameter for the success of the HiHTP algorithm and hence our protocol, we
tested multiple settings for these criteria. We ran our experiments in a
simulation, mimicking a communication over a wireless channel. Firstly, for
both Alice and Bob, normalized $k$-sparse signals $\beta_{A/B}\in
\mathbb{C}^{n}$ were drawn with the locations of the non-zeros distributed
uniformly and entries drawn from the standard normal distribution. For the
signal dimension $n$, we chose \textit{i)} 128 and \textit{ii)} 200
intuitively for two different experiment runs. The channel dimension $\mu$ was
set to \textit{i)} 100 and \textit{ii)} 160, while the $s$ channel
coefficients were also drawn randomly at random supports, resulting in a
$s$-sparse channel $h\in\mathbb{C}^{\mu}$. As the channel is reciprocal
between two communicating partners, we used this channel model for both
$h_{A\rightarrow B}=h_{B\rightarrow A}=h_{AB}$. Additionally, we randomly drew
a codebook $Q\in\mathbb{C}^{\mu\times n}$. For experimental purposes, we did
not fix the parameters $k$ (signal sparsity) and $s$ (channel sparsity) to
specific values, such that we could elaborate the effect of changing these
parameters on the resulting shared secret. In fact, the signal sparsity is a
parameter, which the users can independently choose anyway, whereas the
channel sparsity is given by nature. Although the channel sparsity is not a
parameter one can directly influence in real scenarios, it affects the success
probability of the blind deconvolution. To simulate the signal exchange over
the channel, we applied the previously introduced \textit{circular
convolution} between the channel $h_{AB}$ and $Q\beta_{A/B}$ to both sides
separately and used the result $y_{A/B}=h_{AB}\ast Q\beta_{B/A}+n_{B/A}$ as an
input to the \textit{HiHTP} algorithm at the respectively other party. With
$h_{AB}\otimes\beta_{A/B}$ being the output of the algorithm, we do not need
to separate the two factors (e.g. via SVD), as we directly use this solution
to compute the secret $\mathfrak{c}$. As a metric for the success of our
protocol, we considered the Root Mean Squared Error (RMSE) between the
normalized computed secrets of Alice and Bob. Ideally, we want to achieve a
low RMSE, meaning that Alice's and Bob's computed secrets are indeed
(approximately) equal to an extent.

\begin{figure}[ptb]
\centering
\begin{subfigure}[b]{0.24\textwidth}
\centering
\includegraphics[width=0.7\linewidth]{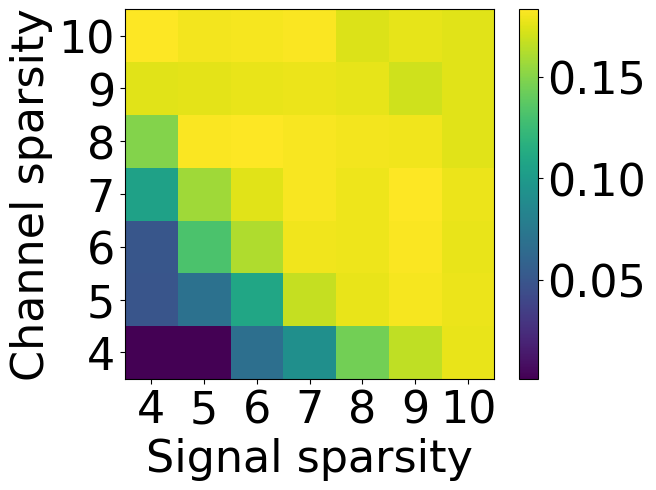}
\caption{Signal dimension $n=128$, channel dimension $\mu=100$}
\label{fig:simulation_channel}
\end{subfigure}
\begin{subfigure}[b]{0.24\textwidth}
\centering
\includegraphics[width=0.7\linewidth]{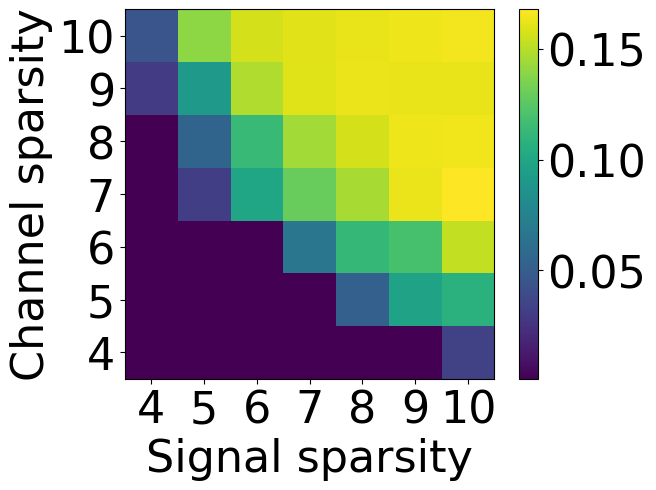}
\caption{Signal dimension $n=200$, channel dimension $\mu=160$}
\label{fig:simulation_dimensions}
\end{subfigure}
\caption{RMSE of Alice's and Bob's shared secret under different settings for
signal and channel sparsity.}%
\label{fig:simulation_sparsities}%
\end{figure}

We continued by defining a range of values for $k$ and $s$, starting at 4 up
to 10. In addition, we used a Signal-to-noise (SNR) ratio of $30$~dB scaling
randomly drawn noise vectors that are added to the transmitted signal to
simulate Additive White Gaussian Noise (AWGN). The SNR quantifies the amount
of noise in relation to the signal and depends on channel circumstances, such
as interference or signal fading effect. With every combination of signal and
channel sparsity, we ran the simulation 50 times, computed the RMSE between
the two resulting shared secrets, normalized between 0 and 1, and stored its
average over these 50 runs. It is to note that in every single run, the
signals $\beta_{A/B}$, the channel $h_{AB}$, as well as the codebook $Q$ were
chosen randomly. Figure~\ref{fig:simulation_channel} depicts the RMSE across
different signal and channel sparsity settings. It becomes clear that with
more non-zero entries in one component, i.e. channel or signal, the RMSE rises
rapidly and saturates at around 0.175 as the similarity of Alice's and Bob's
computed key shrinks. This of course also relies on the dimensions of the
signal and channel. With higher dimensions, the protocol succeeds also with
more non-zero entries, as seen in Figure~\ref{fig:simulation_dimensions} for
signal dimension $n=200$ and channel dimension $\mu=160$. In fact, the success
of our protocol relies on the ability of the HiHTP algorithm to correctly
deconvolve the signal input and recover $h\otimes\beta_{A/B}$, which is given
with a high probability if both signals are sparse respective to their dimension.

\begin{figure}[ptb]
\centering
\includegraphics[width=0.6\linewidth]{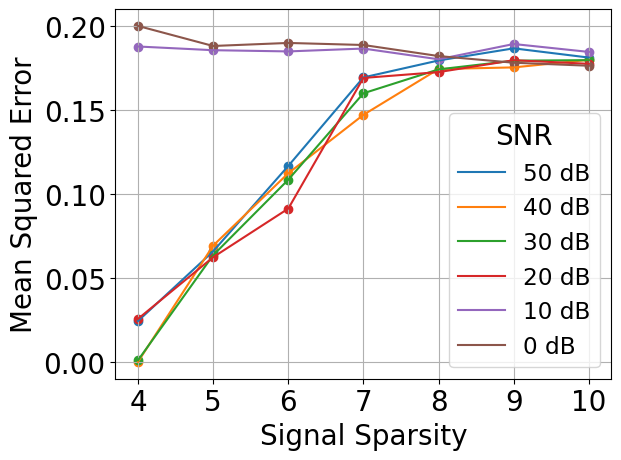}\caption{RMSE of Alice's
and Bob's shared secret under different settings for SNR and signal sparsity.}%
\label{fig:simulation_noise}%
\end{figure}

\subsection{Noise effects}

In wireless communication scenarios, we need to cope with the existence of
noise, which is added to the received signal. In a second experiment, we test
the influence of noise on the functionality of our protocol. The randomly
drawn noise vector is scaled by factors of $10^{-5}$ to $10^{-0}$ and added to
the transmitted signal, representing a SNR of 50~dB (low noise power) to 0~dB
(same noise power as signal power). In addition, we fixed the signal dimension
to $n=128$, the channel dimension to $\mu=100$ and the channel sparsity to
$s=5$ for this test. Figure~\ref{fig:simulation_noise} shows the RMSE of
different settings for signal sparsity and SNR ratio.
As expected, with a higher noise levels, the RMSE is high even for a signal
sparsity of 4. Especially when the SNR is $0$~dB or $10$~dB, the computed
secrets differ significantly and the key generation method fails.
Nevertheless, the results for a SNR from $20$~dB upwards show that the
protocol works even with a notable amount of noise.

\subsection{Signal recovery attack}

\label{sec:signalattack}

\begin{figure}[ptb]
\centering
\includegraphics[width=0.7\linewidth]{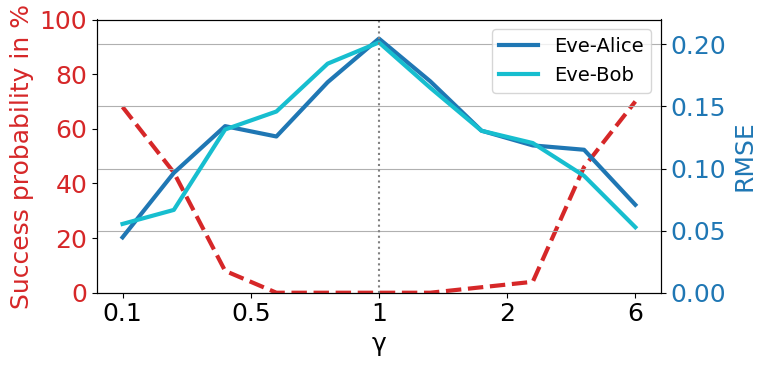} \caption{Success
probability of the attack described in Section \ref{sec:signalattack} (red)
and RMSE between Eve's computed key and Alice's (blue) and Bob's (cyan) key.
The attack was counted as successful, if the absolute element-wise difference
did not exceed $10^{-4}$.}%
\label{fig:attack}%
\end{figure}

As described in Section~\ref{sec:wdh_scheme}, we assume that Eve receives a
superposition of Alice's and Bob's signal $Q(\beta_{A} + \gamma\beta_{B})$
(after the deconvolution process), with respect to the strength of the channel
to either side. The weight is balanced by the factor $\gamma$, describing the
power difference of Alice's and Bob's signal in the superposition Eve
observes. Hence, with $\gamma= 1$, the amplitudes of the superposed signals
are equal, whereas a deviation to either side leads to slight dominance of one
side's signal over the other.

Let us first evaluate the ability of Eve to separate Alice's and Bob's
contributions to the received superposition and thus being able to compute the
shared secret in the noiseless case, by empirically evaluating the following
attack:

\begin{enumerate}
[(i)]

\item Eve solves the blind deconvolution problem
\[
y_{E}=h_{AB}\ast Q(\beta_{A}+\gamma\beta_{B})
\]
for $h_{AB}\otimes(\beta_{A}+\gamma\beta_{B})$. She can again do so using
HiHTP, but with a signal sparsity parameter of $2k$ instead of just $k$, in
order to receive all $2k$ signal peaks and not just $k$ from one party.

\item She now calculates an SVD of $h_{AB}\otimes(\beta_{A}+\gamma\beta_{B})$,
to recover $h_{AB}$ and $\beta_{A}+\gamma\beta_{B}$.

\item She now sorts the elements of $\beta_{A}+\gamma\beta_{B}$ in descending
amplitude order and assigns the $k$ highest values to one party and the $k$
next highest to the other.
\end{enumerate}

After this step, Eve has successfully reconstructed the signals $\beta_{A}$
and $\beta_{B}$ and has now every component necessary to compute the shared
secret, i.e. $\beta_{A}$, $\beta_{B}$ and $h_{AB}$. She can now compute the
secret as $\mathfrak{c}^{\prime\prime}=\widehat{\mathrm{vec}(h_{AB}^{\Uparrow
}\otimes\beta_{B})}\cdot\widehat{\mathrm{vec}(e_{0}^{\mu}\otimes\beta_{A}%
)}=\widehat{h_{AB}^{\Uparrow}}\cdot\widehat{\beta_{A}^{\uparrow}}\cdot
\widehat{\beta_{B}^{\uparrow}}$ or vice versa.

In order to elaborate the fatality of this attack, we implement it and test it
for different values for $\gamma$. We want to figure out, to what extent the
parameter $\gamma$ influences the ability to compute the shared secret by
applying the signal recovery attack. For this, we defined values for $\gamma$
in a range of 0.1 up to 6 and let Eve try the attack. The parameters $Q$,
$h_{AB}$, $\beta_{A}$ and $\beta_{B}$ are generated just as before, and the
SNR was set to 50~dB to guarantee that Alice and Bob can agree on a shared
secret in the first place. Figure~\ref{fig:attack} visualizes the success
probability of the attack, as well as the RMSE between Eve's computed key and
Alice's and Bob's key. For values of $\gamma$ that significantly deviate from
1, the attack succeeds, as Eve can recover the correct shared signal, visible
by the high recovery probability of up to 75\% and the low RMSE to Alice's and
Bob's signal. However, for values $\gamma\approx1$, Eve cannot recover the
secret by distinguishing the peaks of the respective signals of Alice and Bob.
In general, we can conclude that the protocol is not vulnerable against this
signal distinguishing attack, as long as $\gamma\approx1$ is given with a
tolerance of up to $0.5<\gamma<2$.

Next, we consider the noisy scenario outlined in Section \ref{sec:noisy}.
Here, Eve receives a signal in the form of $y_{E}=h_{A\rightarrow E}\ast
Q\beta_{A}+\gamma(h_{B\rightarrow E}\ast Q\beta_{B}$), where $h_{A\rightarrow
E}=h_{AB}+n_{A}$ and $h_{B\rightarrow E}=h_{AB}+n_{B}$. Here, $n_{A/B}$ is
AWGN that distorts the respective channel $h_{A\rightarrow E}$ or
$h_{B\rightarrow E}$ to make it different from $h_{AB}$.

\begin{figure}[ptb]
\centering
\begin{subfigure}[b]{0.24\textwidth}
\centering
\includegraphics[width=0.7\linewidth]{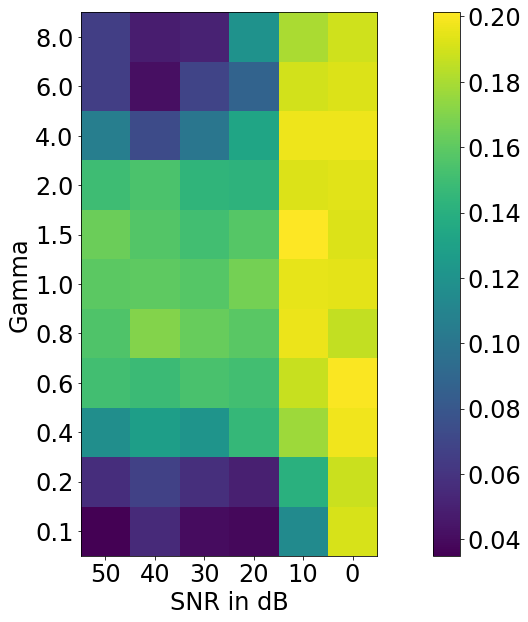}
\caption{One channel is the same.}
\label{fig:diff_channels}
\end{subfigure}
\begin{subfigure}[b]{0.24\textwidth}
\centering
\includegraphics[width=0.7\linewidth]{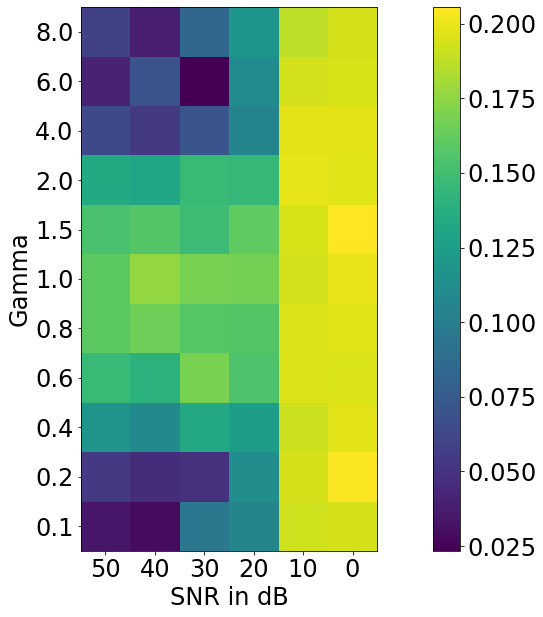}
\caption{Both channels vary.}
\label{fig:diff_channels2}
\end{subfigure}
\caption{MSE of Eve's secret to Alice's and Bob's secret.}%
\label{fig:channel_attack}%
\end{figure}

Figure~\ref{fig:channel_attack} depicts the results of this extended attack
scenario, where in the experiment of Figure~\ref{fig:diff_channels} one of the
two attacker channels was equal to $h_{AB}$ and in
Figure~\ref{fig:diff_channels2} both were different to an extent. The channel
difference was controlled by different noise levels from 50 dB to 0 dB with 0
dB being the case where the noise level equals the channel power. We see that,
as hypothesized in the theoretical section, the recovery of the signals is not
dramatically improved by assuming the more general scenario
\eqref{eq:erweitert} for small $n_{A}$, $n_{B}$. On the other hand, if $n_{A}$
and $n_{B}$ are big (SNR$\leq$10 dB), Eve can hardly obtain any information
about the secret.

In summary, we conclude that Eve can only obtain information about the shared
secret, if she has access to \textit{i)} the same channel as Alice and Bob and
\textit{ii)} can assume that the signal strength of one party is at least
twice as strong opposed to the other party ($\gamma< 0.5$ or $\gamma> 2$),
which is obviously the worst case for Alice and Bob. In a more realistic
scenario, where Eve's channels differ significantly from Alice's and Bob's
channels and their signal strengths are approximately equal, Eve cannot
compute the same shared secret as Alice and Bob do, and therefore fails to
attack this key generation scheme.

\section{Conclusion}

In this paper we presented a novel key generation scheme, which allows Alice
and Bob to agree on a secret key over a physical wireless channel in an
ad hoc fashion, i.e. without further control signals etc. This is achieved by
exchanging their individual own chosen signals in a full duplex transmission and
recovering the signal of the other participant via a blind deconvolution
algorithm. With the recovered signal of the other party and the knowledge of
their own signal, Alice and Bob can compute a common secret. Eve however only
observes a superposition of Alice's and Bob's signals from channel
measurements and we prove analytically and experimentally that she cannot gain
any knowledge about the secret key, even under small deviations of the channels to Eve.
We quantify the remaining entropy given Eve's observations and derive a lower bound
for the actual key size.
The particularity of our scheme is that, other than various other
channel-based key establishment schemes, it does not rely on the channel's
entropy and thus is not constrained by the quality of the channel. We
simulated the FD-BBD protocol and verified that with sufficient channel and
signal sparsity, the signals that Alice and Bob compute separately indeed
result in a shared secret. Lastly, we note that the resulting shared secret
eventually needs to be quantized to bits and fed into a key derivation
function, in order to use it as a key for symmetric cryptographic schemes. We
desist from using information reconciliation techniques to cancel out
remaining noise, as this would leak additional information about the shared
secret to the attacker. Future work includes tests with real channel
measurements, where the influence of noise and channel fading can be
elaborated in more detail.

\subsection*{Acknowledgement}

{AF was supported by the Wallenberg AI, Autonomous Systems and Software
Program (WASP) funded by the Knut and Alice Wallenberg Foundation. GW, DB and
BG were supported by the German Science Foundation (DFG) under grants WU 598/8-1, and
WU 598/12-1, the 6G research cluster (6g-ric.de) supported by the German Ministry of
Education and Research (BMBF) in the program 'Souver\"{a}n. Digital. Vernetzt.',
joint project 6G-RIC, project identification number 16KISK020K, as well as the
BMBF project 'Physical Layer Security for E2E IoT Security' under number 16KIS1473.
}

\bibliographystyle{IEEEtran}
\bibliography{wdh}

\appendix

\subsection{Proof of Theorem \ref{theo:underlying_entropy}}

\label{app:deltaproof}

\begin{proof}
The signals $\beta_{A},\beta_{B}$ can be re-written as%
\[
\beta_{A}=\sum_{i\in\{0,...,k-1\}}\alpha_{i}e_{\sigma\left(  i\right)  },.
\]%
\[
\beta_{B}=\sum_{i\in\{k,...,2k-1\}}\alpha_{i}e_{\bar{\sigma}\left(  i\right)
}%
\]
where $\bar{\sigma}=\Sigma\backslash\sigma$ where we have appropriately mapped
the $\alpha_{A},\alpha_{B}$ to $\alpha\in\lbrack0,1]^{2k}$ for simplicity. Let
$i$ be the $i$-th ordering (permutation) such that $\alpha_{\lbrack0]_{i}}%
\leq...\leq\alpha_{\lbrack2k-1]_{i}}$ which defines the random mapping
\[
\alpha_{0},...,\alpha_{2k-1}\hookrightarrow\{1,...,(2k)!\}
\]
The Maximum Likelihood (ML) estimate selects the permutation which orders the
$\alpha_{i}$'s according to their absolute values. Obviously, for $\gamma=1$,
the ML estimate yields the correct permutation with probability $\frac
{k!k!}{(2k)!}$ which is the random choice without any prior knowledge. For the
case $\gamma<1$ we can calculate
\begin{align*}
\mathbb{P}\left(  \max_{0\leq i\leq k-1}\alpha_{i}>x\right)   &
=1-\mathbb{P}\left(  \alpha_{1}\leq x,...,\alpha_{k}\leq x\right) \\
&  =1-x^{k}%
\end{align*}
and%
\begin{align*}
\mathbb{P}\left(  \min_{k\leq i\leq2k-1}\alpha_{i}\leq x\right)   &
=1-\mathbb{P}\left(  \alpha_{k+1}>x,...,\alpha_{2k}>x\right) \\
&  =1-\left(  1-x\right)  ^{k}.
\end{align*}
Hence, we have%
\begin{align*}
&  \mathbb{P}\left(  \max_{0\leq i\leq k-1}\alpha_{i}\leq\min_{k\leq
i\leq2k-1}\alpha_{i}\right) \\
&  =\int_{x}^{1}\mathbb{P}\left(  \max_{0\leq i\leq k-1}\alpha_{i}>x\right)
f_{X_{\min}}dx\\
&  =\int_{x}^{1}\mathbb{P}\left(  \max_{0\leq i\leq k-1}\alpha_{i}>x\right)
\frac{d\left(  1-\left(  1-x\right)  ^{k}\right)  }{dx}dx\\
&  =k\int_{x}^{1}x^{k}\left(  1-x\right)  ^{k-1}dx.
\end{align*}
Obviously, this yields for $\gamma=0$%
\[
\mathbb{P}\left(  \max_{0\leq i\leq k-1}\alpha_{i}\leq\min_{k\leq i\leq
2k-1}\alpha_{i}\right)  =\frac{k!k!}{(2k)!}=\binom{2k}{k}.
\]
Set $\gamma=1-\delta$ then%
\[
\mathbb{P}\left(  \max_{0\leq i\leq k-1}\alpha_{i}>\frac{x}{1-\delta}\right)
=1-\min\left\{  1,\frac{x}{1-\delta}\right\}  ^{k}%
\]
so that%
\begin{align*}
&  \mathbb{P}\left(  \left(  1-\delta\right)  \max_{0\leq i\leq k-1}\alpha
_{i}\leq\min_{k\leq i\leq2k-1}\alpha_{i}\right) \\
&  =k\int_{0}^{1}\min\left\{  1,\frac{x}{1-\delta}\right\}  ^{k}\left(
1-x\right)  ^{k-1}dx\\
&  =\frac{k}{\left(  1-\delta\right)  ^{k}}\int_{0}^{1-\delta}x^{k}\left(
1-x\right)  ^{k-1}dx\\
&  +k\int_{1-\delta}^{1}\left(  1-x\right)  ^{k-1}dx
\end{align*}
The integral can be calculated as%
\begin{align*}
&  \int x^{k}\left(  1-x\right)  ^{k-1}dx\\
&  =-\frac{1}{k}x^{k}\left(  1-x\right)  ^{k}+\frac{k}{k}\int x^{k-1}\left(
1-x\right)  ^{k}dx\\
&  =-\frac{1}{k}x^{k}\left(  1-x\right)  ^{k}-\frac{k}{k\left(  k+1\right)
}x^{k-1}\left(  1-x\right)  ^{k+1}\\
&  +\frac{k\left(  k-1\right)  }{k\left(  k+1\right)  }\int x^{k-2}\left(
1-x\right)  ^{k+1}dx\\
&  =...\\
&  =-\sum_{i=0}^{k-1}\frac{\left(  k-1\right)  \cdots\left(  k-1-i\right)
}{k\cdots\left(  k+i\right)  }x^{k-i}\left(  1-x\right)  ^{k+i}\\
&  +\frac{k!}{k\left(  k+1\right)  \cdots\left(  2k-1\right)  }\int\left(
1-x\right)  ^{2k-1}dx.
\end{align*}
The probability can be bounded as%
\begin{align*}
&  \mathbb{P}\left(  \left(  1-\delta\right)  \max_{0\leq i\leq k-1}\alpha
_{i}\leq\min_{k\leq i\leq2k-1}\alpha_{i}\right) \\
&  =\frac{k}{\left(  1-\delta\right)  ^{k}}\int_{0}^{1-\delta}x^{k}\left(
1-x\right)  ^{k-1}dx\\
&  +k\int_{1-\delta}^{1}\left(  1-x\right)  ^{k-1}dx\\
&  \leq\binom{2k}{k}^{-1}\frac{1-\delta^{2k}}{\left(  1-\delta\right)  ^{k}%
}+\delta^{k}.
\end{align*}
Since the Shannon entropy is bounded below by the min-entropy, the final
result is obtained as%
\begin{align*}
H(\beta_{A}|\beta_{A}+\gamma\beta_{B})  &  \geq H_{\infty}(\beta_{A}|\beta
_{A}+\left(  1-\delta\right)  \beta_{B})\\
&  \geq-\log\left(  \binom{2k}{k}^{-1}\frac{1-\delta^{2k}}{\left(
1-\delta\right)  ^{k}}+\delta^{k}\right)  .
\end{align*}

\end{proof}

\end{document}